\documentclass[sigplan,10pt,nonacm]{acmart}

\usepackage{graphicx}
\usepackage{booktabs}
\usepackage{amsmath}
\usepackage{enumitem}

\setcopyright{none}
\renewcommand\footnotetextcopyrightpermission[1]{}
\AtBeginDocument{\typeout{PROBE tfs=\the\textfloatsep\space bcs=\the\belowcaptionskip\space acs=\the\abovecaptionskip\space fs=\the\floatsep}}

\usepackage{amsthm}
\newtheorem{proposition}{Proposition}
\theoremstyle{definition}
\newtheorem{definition}{Definition}

\newcommand{\needload}{\texttt{need\_to\_load}}
\newcommand{\parhead}[1]{\par\vspace{1.5pt}\noindent\textbf{#1}\enspace}

\begin{document}

\title{Preserving Provenance in Shared KV Caches for LLM Serving}

\author{%
  Wei Song\textsuperscript{1}\quad
  Yuxin Cao\textsuperscript{2}\quad
  Xi Zheng\textsuperscript{3}\quad
  Leo Zhang\textsuperscript{1}\quad
  Xiao Cheng\textsuperscript{3}
}

\affiliation{%
  \institution{\textsuperscript{1}Griffith University}
  \country{}
}

\affiliation{%
  \institution{\textsuperscript{2}National University of Singapore}
  \country{}
}

\affiliation{%
  \institution{\textsuperscript{3}Macquarie University}
  \country{}
}

\renewcommand{\shortauthors}{Song et al.}

\begin{abstract}
Production LLM serving stacks combine an inference engine's local prefix cache with a shared KV-cache tier for fleet-wide reuse. The local cache distinguishes requests by adapter, weight configuration and sharing domain, but the shared tier may key entries only by token content and coarse model metadata. This boundary erases provenance and lets identical tokens under incompatible computational or sharing contexts collide. We call this composition gap \emph{provenance-blind reuse} and present its first systematic study. A source audit of three vLLM connectors confirms the structural omission, while runtime experiments reproduce it across vLLM and two SGLang releases, 12 models from 7 families (0.5\,B--32\,B), and over 160 configurations. Cross-adapter collisions reduce accuracy from 0.94 to 0.64, incompatible KV representations reduce reasoning accuracy to zero, and salt omission enables 93\% prompt identification from timing. We formalize the missing guarantee as the \emph{KV provenance contract}: for a declared dimension registry, shared keys must be injective over computational and sharing provenance, with identities stable across workers. Any dimension with a stable identity can therefore be added without connector-specific key logic. A canonical descriptor binds per-request and per-worker provenance into lookup and store keys, while a differential checker detects dimensions that change KV state without changing the key. Implemented in vLLM and SGLang~0.5.20 across three cache paths, provenance binding eliminates unsafe reuse while preserving legitimate sharing. Hit-path latency changes remain within 0.34~ms and below run-to-run variation; retention grows with provenance diversity.
\end{abstract}

\maketitle

\section{Introduction}

\suppressfloats[t]
\begin{figure}[t]
  \centering
  \includegraphics[width=0.92\columnwidth]{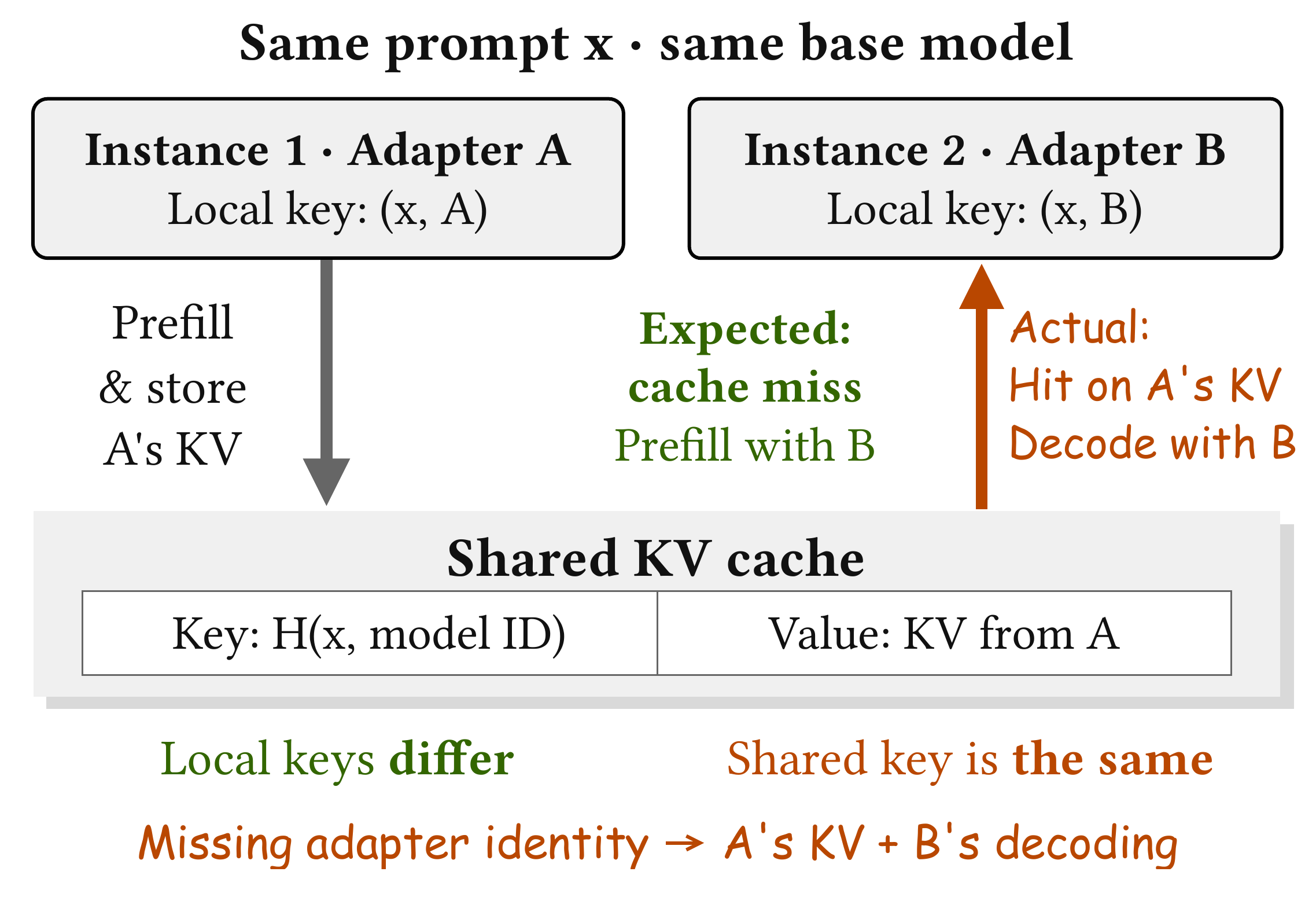}
  \caption{Cross-adapter provenance-blind reuse on vLLM. Because the shared key omits adapter identity, a request under~$B$ retrieves KV state produced under~$A$ and decodes with $B$'s weights over $A$'s prefix representation.}
  \label{fig:cross-adapter-example}
\end{figure}

Modern LLM serving systems increasingly rely on two cooperating layers of KV-cache reuse. An inference engine such as vLLM~\cite{kwon2023efficient} or SGLang~\cite{zheng2024sglang} manages a local prefix cache and distinguishes requests by their serving context, including adapter identity, weight configuration and, where supported, a sharing-domain salt. A shared tier such as LMCache~\cite{cheng2025lmcache} or Mooncake~\cite{qin2025mooncake} serves computed KV state to any worker whose request maps to the same key, enabling fleet-wide prefix reuse. Each layer is correct in isolation, but their composition is safe only if the shared key preserves every engine-level distinction that affects the cached state or permission to share it. As we show, this condition is routinely violated because context tracked locally by the engine is lost at the engine--connector boundary.

Figure~\ref{fig:cross-adapter-example} illustrates the failure concretely.
Two workers serve the same base model with different LoRA adapters~\cite{hu2022lora}, $A$ and $B$, while sharing an external KV tier.
The first worker computes KV state for a prompt under adapter~$A$ and stores it under a key derived from the prompt tokens.
When the same prompt later arrives at the second worker under adapter~$B$, the shared tier matches on the token-derived key and returns the state that $A$ produced.
The second worker therefore decodes with $B$'s weights over a prefix representation computed under~$A$, producing output that silently depends on an adapter the request did not select.
No error is raised and no adversary is needed. The engine's local cache would have kept the two adapters separate.
The failure arises solely from composing that local isolation with a shared key that omits adapter identity, which the engine alone holds.

The root cause generalizes beyond adapters.
The engine and the shared-cache connector derive their keys under different notions of equivalence: the engine distinguishes two identical token sequences whenever their request contexts produce different KV state or prohibit sharing, whereas the connector treats them as equivalent based on tokens and coarse model metadata alone.
Equality of shared keys thus implies neither computational equivalence nor permission to reuse.
We call this gap \emph{provenance-blind reuse}.
While adapter-related key conflicts have surfaced individually as engineering issues~\cite{lmcacheissue,vllmissue}, our work unifies them with failures involving weight precision, weight version, KV representation format, and sharing domains as instances of a single missing interface contract at the engine--connector boundary.

We conduct a systematic study of this failure across two inference engines, multiple connector implementations, and a wide range of model configurations.
A source-level audit of three independently implemented shared-cache connectors for vLLM (LMCache, the filesystem connector, and the Mooncake integration) confirms that the key omission is structural. All three derive keys the same way, a consequence of content addressing rather than a per-connector defect.
Runtime experiments reproduce the failure through vLLM's LMCache and filesystem paths and SGLang's HiCache path across two releases, spanning 12 models from 7 families (0.5\,B--32\,B) and more than 160 configurations.
The exact omitted dimensions vary by stack and release, but the pattern is consistent: the shared tier never receives the full provenance needed to decide whether cached state is safely reusable.

We organize the missing provenance into two classes with distinct failure modes. \emph{Computational provenance}, such as adapter identity, weight precision, weight version and KV representation, determines whether cached state is \emph{valid} for the requester; its omission can cause the shared tier to serve or load incompatible state. \emph{Sharing provenance}, represented by the tenant salt, determines whether otherwise compatible requests are \emph{permitted} to share; its omission defeats engine-level tenant isolation. The consequences vary: cross-adapter reuse reduces accuracy from 0.94 to 0.64 without an adversary, while incompatible KV representations reduce reasoning accuracy from 0.94 to zero. Weight-version harm is model- and task-dependent, whereas weight-precision reuse alters outputs without a statistically significant accuracy loss in the workloads tested. Salt omission lets a remote observer identify a victim's prompt with 93\% accuracy from timing alone. Routing locality and prefix diversity affect how often workers reach the shared tier but do not remove the flaw.

The most direct repair is to include the missing dimensions in the shared key.
However, two fundamental constraints make such a naive approach silently break reuse or leave the gap open.
First, process-local identifiers are not provenance identities. Two workers may assign different integer handles to the same LoRA adapter, so binding the raw handle would partition state that should be shared.
Second, key derivation occurs at execution sites with asymmetric metadata. In vLLM, for example, the scheduler that looks up the store knows the request's adapter and sharing domain, while the worker that writes the state may not. Injecting provenance at only one site therefore causes the two paths to derive different keys.

To address both constraints, we propose the \emph{KV provenance contract}, an interface specification parametric over a declared dimension registry at the engine--connector boundary. It requires configurations that produce different KV state to derive different shared keys (computational invariant), and configurations not permitted to share to do likewise (permission invariant). Both invariants require provenance identities to remain stable across workers and consistent between lookup and store paths. Extensibility follows structurally from the parametric design: any dimension with an injective, cross-worker-stable identity extends the contract without changing connector-specific key logic. We realize the contract with a canonically encoded provenance descriptor spanning two lifetimes. Per-request fields carry the adapter identity and salt through the request record to both access paths, while per-worker fields capture weight precision, version and KV format at startup. Each cache path's native key derivation consumes this descriptor. A differential checker holds token content fixed and reports any declared dimension that changes KV state without changing the derived key. We implement provenance binding in vLLM and SGLang~0.5.20 across LMCache, vLLM's filesystem connector and SGLang HiCache, eliminating cross-provenance reuse while fully preserving same-provenance sharing. Across both engines, observed hit-path latency changes are at most 0.34~ms and smaller than run-to-run variation, so no hit-path latency penalty is resolvable at our measurement scale; retention grows with provenance diversity on the file-backed store that each engine uses.

In summary, this paper makes three contributions:
\begin{itemize}[leftmargin=1em,labelsep=0.4em, topsep=.3em]

\item \textbf{First systematic identification of provenance-blind reuse.}
We present the first analysis of the composition gap between an LLM engine's
local cache and a shared KV-cache tier. We show that adapter, weight-precision,
weight-version, KV-representation, and sharing-domain conflicts, previously
reported as isolated bugs, are instances of a single missing interface contract
at the engine--connector boundary. A source-level audit covers three independently
implemented connectors for vLLM, and runtime experiments confirm the failure
across two engines (vLLM and SGLang), two SGLang releases, 12 models from
7 families (0.5\,B--32\,B), and more than 160 configurations.

\item \textbf{The KV provenance contract.}
We formalize the \emph{KV provenance contract}, a registry-parametric interface specifying when two requests may safely share cached state. It distinguishes computational from sharing provenance and requires each dimension affecting KV state or sharing permission to carry a stable, cross-worker identity. Extensibility is a formal property of the design: any new stable-identity dimension extends the contract without changing connector-specific key logic. We instantiate the contract for five dimensions: adapter, precision, version, KV format and tenant salt. A per-dimension differential omission test falsifies the contract under declared configuration changes, making omitted bindings directly testable.

\item \textbf{Cross-engine implementation of the KV provenance contract.}
We implement the contract as provenance binding in vLLM and SGLang~0.5.20 across three cache paths (LMCache, vLLM's filesystem connector, and SGLang HiCache). Across all paths it eliminates cross-provenance reuse while preserving same-provenance sharing, with no resolvable hit-path latency penalty at our measurement scale and storage cost linear in provenance diversity.

\end{itemize}

\section{Background and Threat Model}
\label{sec:background}

\subsection{Local and Shared KV Caching}
\label{sec:kvcaching}

During prefill, an autoregressive transformer computes key and value projections for every prompt token at every model layer. Subsequent decoding attends to these projections rather than recomputing the prompt. We call their collection the request's \emph{KV state}. KV state consumes substantial serving memory, and producing it accounts for much of the time to first token, so requests sharing a token prefix can avoid redundant prefill by reusing the state. A \emph{local prefix cache} retains that state within one inference worker and indexes it by the shared prefix: vLLM~\cite{kwon2023efficient} hashes fixed-size token blocks, while SGLang~\cite{zheng2024sglang} organizes prefixes in a radix tree. Because the cached state remains local to the worker, only requests routed to that worker can reuse it.

A \emph{shared KV-cache tier} extends this reuse across workers by placing cached state in process-external storage. Systems such as LMCache~\cite{cheng2025lmcache} and Mooncake~\cite{qin2025mooncake}, vLLM's filesystem connector, and SGLang HiCache expose such storage through a connector at the engine boundary. On a local miss, the connector derives a shared key from the prefix and available metadata and queries the external store; a hit loads the cached tensors and skips the corresponding prefill work, while a miss computes and writes new state. Because the store outlives individual workers, state written by one remains available to another across scheduling and deployment changes. The connector's shared key (not the engine's local-cache key) therefore determines which requests are treated as equivalent at fleet scope.

\subsection{Serving Configurations and Isolation}
\label{sec:configs}

Safe KV reuse requires computational compatibility and permission to share. \emph{Computational provenance} determines validity for the requesting configuration, while \emph{sharing provenance} determines whether reuse is allowed.

\subsubsection{Computational configuration}
We study four dimensions, each of which can change the KV state that is produced for an identical token prefix.

\parhead{Adapter identity.}
A LoRA adapter~\cite{hu2022lora} applies learned low-rank updates to a base model's weights. Because the key and value projections depend on the active weights, the same token prefix can produce different KV state under different adapters. Multi-adapter systems such as S-LoRA~\cite{sheng2024slora} and Punica~\cite{chen2024punica} multiplex adapters for different tasks or tenants over one base model, so the adapter, and hence the KV state, can vary between requests on the same worker.

\parhead{Weight precision.}
Post-training quantization~\cite{frantar2023gptq,lin2024awq,xiao2023smoothquant} reduces model-memory and compute costs by representing or executing the weights at lower precision. A serving fleet may therefore run unquantized, \texttt{fp8} or lower-bit replicas of the same logical model. These replicas can produce numerically different key and value projections even when the KV cache itself uses the same element type. Weight precision is normally fixed for a worker's lifetime.

\parhead{Weight version.}
Model updates replace a checkpoint behind a stable model name. During a rolling update, workers running the old and new checkpoints may coexist and access the same shared cache. Even when the checkpoints have the same architecture and tensor shapes, their parameters can produce different KV state for an identical prefix. Weight version is fixed for the lifetime of a worker, and it changes only across deployment events.

\parhead{KV representation.}
The stored tensors also have a worker-level representation. An engine may allocate its KV pool with different element types, such as \texttt{bfloat16} or \texttt{float16}, or apply a KV-specific format such as \texttt{fp8}. Unlike weight precision, which affects the projections themselves, the KV representation governs how computed state is stored and loaded. State written under one representation cannot be assumed compatible with a worker using another: its values may be misinterpreted or the load may fail. The effective representation is fixed when a worker initializes its KV pool.

\subsubsection{Sharing permission}
Computational compatibility alone does not make reuse permissible.

\parhead{Sharing-domain salt.}
vLLM accepts a per-request \texttt{cache\_salt} and incorporates it into the first block hash of its local prefix cache, documenting this mechanism as a defense against prefix-cache timing attacks (CVE-2025-46570)~\cite{vllm2025timing}. Requests with different salts occupy separate sharing domains and must not reuse one another's state, even when their tokens and computational configurations match. The salt changes not the computed KV state but whether another request may observe or reuse it; like the adapter, it can vary per request within one worker.

Adapter identity and sharing domain are \emph{request-level}, whereas weight precision, weight version and KV representation are \emph{worker-level}. The engine can know all five dimensions, but a shared-cache connector can bind only the provenance it receives. Because local-cache handling varies across engines and releases, local isolation does not imply that an external tier preserves the same distinctions. The interface contract of \S\ref{sec:defense} is shaped by this separation, and it binds the two levels through separate digests.

\subsection{Deployment and Threat Model}
\label{sec:threat}

We study two consequences of the same interface failure: integrity or liveness failures under ordinary operation, and confidentiality leakage across isolated request domains.

\parhead{Integrity and liveness setting (no adversary).}
Clients issue well-formed requests to a serving fleet whose workers share an external KV tier but differ in a dimension omitted from its key: adapter identity, weight precision, weight version or KV representation. Such heterogeneity arises in multi-adapter serving, mixed-precision deployment, rolling model updates and workers configured with different KV formats. A request may then retrieve state produced under a configuration other than the one it selected. Depending on the dimension, the foreign state can perturb the output, reduce task accuracy, or fail to load correctly. No malicious request is required. The failure follows from normal requests and valid worker configurations sharing one namespace.

\parhead{Confidentiality setting (cross-instance observer).}
A victim and an attacker issue requests to different engine instances that share the external KV tier. The victim supplies a sharing-domain salt intended to prevent other requests from probing its cached prefixes. The attacker does not know that salt and observes only the time to first token of its own requests. Given a set of candidate prompts, the attacker tests each candidate from another instance. A shared-tier hit can reveal that the victim previously cached the corresponding prefix. The goal is prompt identification within this candidate set, not arbitrary prompt reconstruction. Unlike prior prefix-cache timing attacks~\cite{inputsnatch,earlybird,promptpeek}, this setting requires no co-location with the victim and tests whether engine-local salt isolation survives composition with the shared tier.

\section{Provenance-Blind Reuse}
\label{sec:erasure}

Safe KV reuse requires the shared key to preserve every distinction made locally by the engine. We formalize this requirement and define provenance-blind reuse as its violation, then audit three independently implemented vLLM shared-cache connectors, all of which derive keys alike and discard the same provenance at the engine--connector boundary, showing that the failure is structural.

\subsection{Conditions for Safe KV Reuse}
\label{sec:safereuse}

Reuse of KV state is safe when the state is \emph{applicable} to the requesting configuration and the request is \emph{permitted} to use it. Both conditions must be checkable from the shared key alone, because the store decides reuse without consulting the engine, and \S\ref{sec:defense} states them as the contract's invariants.

\begin{definition}[Provenance-blind reuse]
\label{def:provblind}
Provenance-blind reuse occurs when two requests generate the same shared cache key because they have the same prompt tokens, even though they differ in configuration or sharing policy. As a result, one request may reuse the KV state produced for the other even when that reuse is invalid or not permitted.
\end{definition}

This problem is not a bug in one connector. It arises because content-addressed caches deliberately hash only prompt tokens to reuse identical prefixes across workers. Information absent from the tokens, such as adapter identity, weight version or sharing policy, is therefore lost unless explicitly included in the shared key.

\subsection{What Information Reaches the Shared Cache}
\label{sec:audit}

We audited every content-addressed connector in a current vLLM and LMCache stack, covering the three common paths for cross-instance sharing: LMCache's connector, vLLM's built-in filesystem store (\texttt{Shared\-Storage\-Connector}) and its Mooncake integration. For each, we compare the engine's context at KV computation with what the connector receives and the shared key contains.
At KV computation, the engine holds the full context: token sequence, adapter identity, salt, and weights at the current precision and version. The connector receives only token identifiers, model name, tensor-parallel world size, worker rank and KV dtype. LMCache also accepts an optional tag, but only through a nonstandard request channel rather than the engine's standard request record. None of the three receives the salt, adapter identity, weight precision distinct from KV dtype, or weight version.

\subsection{How Current Connectors Lose Provenance}
\label{sec:provloss}

Each key reflects only the inputs received. LMCache combines a hash of the token identifiers with the model name, tensor-parallel width, worker rank and KV dtype; vLLM's \texttt{Shared\-Storage\-Connector} uses their \texttt{md5}, whereas its Mooncake connector applies \texttt{blake2b} hash. These fields confine reuse within one nominal model configuration, while multimodal inputs remain safe because their content hashes enter the token identifiers upstream. All three omit adapter identity, weight precision distinct from KV dtype, weight version and tenant salt. KV-representation coverage differs by stack: LMCache already binds KV dtype, while the representation failures in \S\ref{sec:reuse} arise on SGLang, where the shared key omits both the element type and the dtype distinctions that our experiments cover.

The evidence rests on three source audits and two runtime verifications: LMCache and
vLLM's filesystem store also collide end to end, and Mooncake keys its chunks the same
token-content way by source audit. Because the three are independently developed yet
derive their keys alike, the exposure is a property of the content-addressing pattern
rather than an accident in one connector. One shared-store deployment guide instructs
operators to change the store namespace whenever weights, quantization or adapters
change, first-party evidence that the key encodes none of them.

\subsection{Why Requests Collide}
\label{sec:collisionpath}

Figure~\ref{fig:cross-adapter-example} illustrates the concrete case for adapters, and
the mechanism generalizes to every discarded dimension.
Figure~\ref{fig:provenance-boundary} traces the data flow and marks the boundary between
engine and connector at which provenance is lost.

\begin{figure}[t]
  \centering
  \includegraphics[width=0.86\columnwidth]{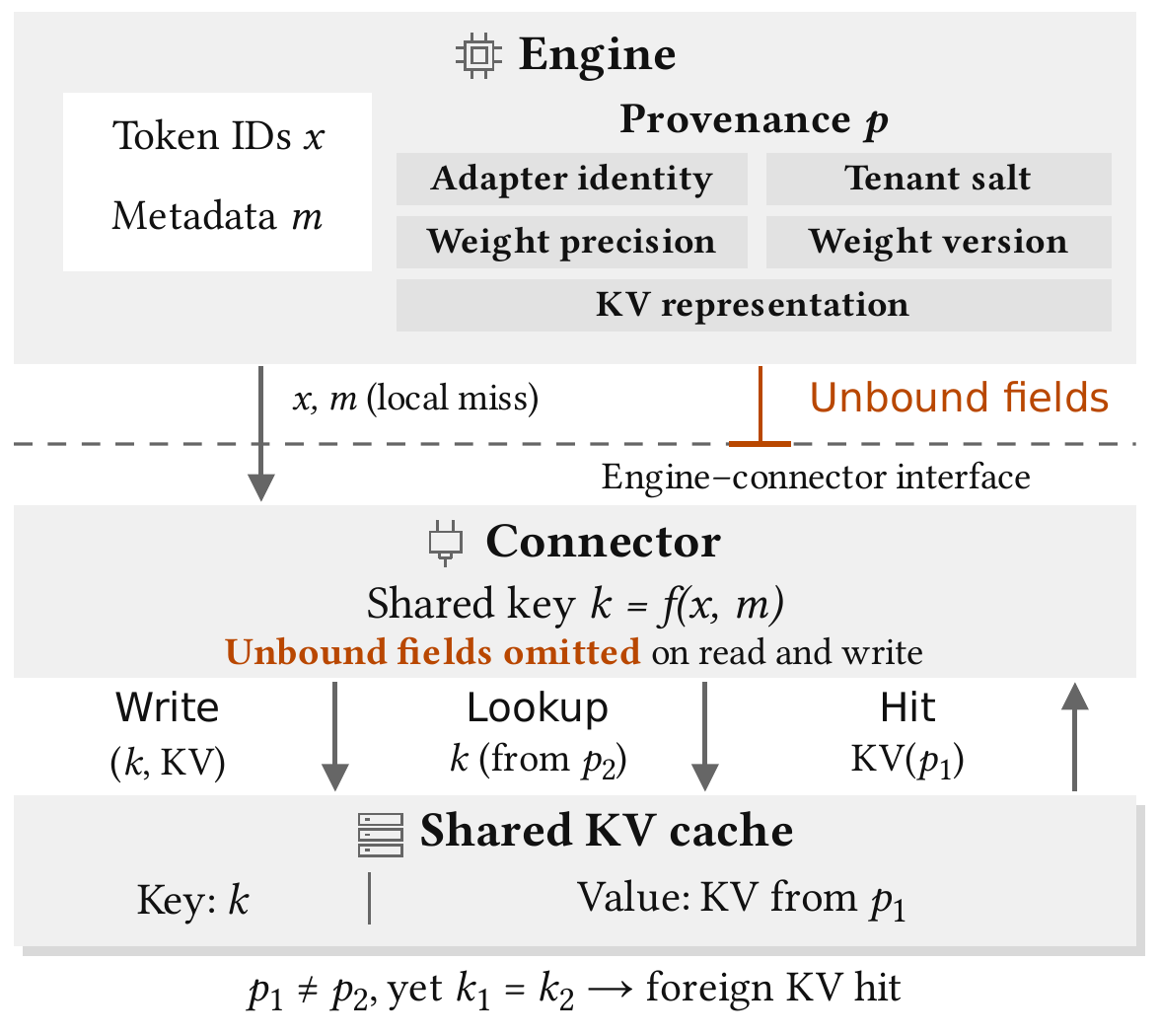}
  \caption{The unrepaired path. Provenance is lost at the engine--connector interface: adapter, salt, precision, version, and KV representation reach neither the read nor the write path of shared-key derivation, so two contexts $p_1 \neq p_2$ derive one key and the second is served the first's state. $m$ summarizes the retained model fields, which vary by connector.}
  \label{fig:provenance-boundary}
\end{figure}

When a request misses the local prefix cache, the engine passes token identifiers to the connector, which derives a key and queries the store. Discarded dimensions cross neither read nor write paths, so requests sharing tokens and nominal model derive one key, and a hit returns the first entry written. The collision lies in the connector interface rather than the store: any store accepting that key inherits ambiguity.

\section{Empirical Characterization}
\label{sec:collisions}

The previous section showed that shared-cache connectors lose provenance by design. We now ask whether cross-provenance reuse occurs, how much it degrades correctness, what causes that degradation, and how likely it becomes under the conditions of a realistic deployment.

\subsection{Characterization Setup}
\label{sec:methodology}

We measure two quantities. \emph{Confirmed cross-provenance reuse} is a nonzero \needload{} counter, an engine-internal metric reporting prefix tokens loaded from the shared store rather than recomputed, when a worker requests under one provenance tokens seeded on another GPU under a different provenance. \emph{Output divergence} is a served output differing from an isolated reference under the same configuration with the tier disabled, excluding decoding nondeterminism. Where ground truth exists, we also grade correctness (\S\ref{sec:accuracy}).

\parhead{Serving stacks and platforms.} We characterize two stacks end to end: \textbf{vLLM~0.9.2 with LMCache} and \textbf{SGLang with HiCache} on releases 0.5.20 and 0.5.8.post1. Experiments use two platforms: P1, with two RTX~6000 Ada GPUs (driver 550.163.01, CUDA~12.4), and P2, with RTX PRO 6000 and RTX PRO 5000 Blackwell GPUs of compute capability 12.0 (driver 580.173.02, CUDA~13.0). vLLM runs on P1 with a standalone cache server over the deployed transport, with scale and family checks on P2. SGLang~0.5.20 runs on P2 and 0.5.8.post1 on P1 because P1's driver cannot load the CUDA~13 build required by later releases. We hold the release fixed across model-size and family cells to avoid confounding with engine version. Correctness comparison remains within one software configuration and platform.

Unless stated otherwise, each experiment pairs a \emph{seeder} populating the store under one provenance with a \emph{retriever} requesting the same tokens under another. We compare the retriever with an \emph{isolated reference} running its configuration with the tier disabled and its own \emph{warm} baseline reading a self-seeded store, separating foreign state from reuse. Each phase runs in a separate server process, attributing effects to state written by a prior process. A cross-provenance change from a correct answer to wrong or vice versa is a \emph{flip}. We test flips with an exact-binomial McNemar test and report them as harmful (correct\,$\to$\,wrong) or reversed (wrong\,$\to$\,correct).

\parhead{Datasets and workloads.} We use public datasets when a task provides ground truth. Reuse and divergence experiments use the first human turn from 2,000 ShareGPT conversations~\cite{sharegpt}, spanning 200--4,000 characters. Incidence and fleet experiments pair each with one of eight hand-written system prompts over 200 tokens as the shared prefix. Sentiment correctness uses GLUE SST-2~\cite{socher2013sst,wang2019glue} for short reviews and IMDB~\cite{maas2011imdb} for long documents. Opposite-policy pairs train on a 300-example sample (150 per class) and evaluate on a disjoint split: the first 400 SST-2 validation sentences (300 for scale sweep and 100 for short-review experiment) and 400 IMDB test reviews (200 per class, fixed seed), using 100, 200 or all 400. IMDB reviews have markup removed and are clipped to 1,400 characters so the review and instruction fit the context. Mathematical reasoning uses 100 GSM8K~\cite{cobbe2021gsm8k} problems with few-shot prompting and exact-match grading. Extraction and constructed reasoning use passages longer than the connector's 128-token chunk with programmatically computable ground truth: 60 synthetic memoranda of about 160 tokens with span-retrieval questions, and 60 or 100 fixed-seed items of about 340 words spanning aggregation, comparison, filter-and-maximum and two-hop lookup.

\subsection{Cross-Provenance Reuse (Does It Happen?)}
\label{sec:reuse}

Cross-provenance reuse occurs on every dimension the storage key
omits, including adapter identity, weight precision, weight version, and
key-value representation, and persists at every model size tested, from
0.5\,B to 32\,B parameters.

Adapter identity is the clearest case. Two independently initialized random adapters collide and diverge on all six prompts. With fine-tuned English and French adapters, all thirty prompts collide and 25 of 30 outputs diverge from the requester's output (83\%, 95\% Wilson interval [66, 93]). A tier-disabled single-instance control keeps the adapters separate, attributing the divergence to the store: the foreign KV representation changes how the model reads the prompt while the requester still decodes. SGLang~0.5.20 partially mitigates this case by adding a per-adapter handle to the storage key (\texttt{uuid5} for launch-time adapters and a fresh \texttt{uuid4} for run-time loads). Over 200 reviews, two opposite-policy adapters stay apart: the requester loads nothing from the other's store, writes 713 entries and matches its isolated reference (0.920 vs.\ 0.920, no flips), while same-adapter reuse loads 29,312 tokens. The source scores 0.050 under the reader's inverted policy, confirming the policies oppose. Yet a name-and-path handle is blind to the bytes it names: replacing weights in place leaves the key unchanged, and the replacement worker receives the earlier adapter's state over the full 5,376-token prefix and diverges from its isolated reference (\S\ref{sec:blocking}). The handle also over-separates, so identical content at another path or dynamically loaded on two workers shares nothing. Under 0.5.8.post1, persistent reuse does not complete with an adapter loaded, leaving adapter isolation unresolved.

Weight precision and version collide through the same mechanism. Dynamic \texttt{fp8} quantization changes weights while KV dtype stays \texttt{bfloat16}, yielding identical keys: all forty prompts collide and 34 outputs diverge (85\%, [71, 93]), from rewording to factual disagreement. On vLLM's \texttt{Shared\-Storage\-Connector}, 16 of 20 prompts collide and four misses fall below the chunk size; the result also appears in SGLang's graded precision cells (\S\ref{sec:accuracy}). Weight version behaves similarly: the key identifies served path rather than content, so an in-place update preserves the key and a populated store serves superseded state to workers running new weights; all ten prompts retrieve that state and diverge. Under SGLang's full controls, changing the checkpoint at one served name consumes 960 of about 995 tokens from the other version's state, versus zero on a clean store, zero under a distinct served name and normal same-provenance reuse. KV pool element type is likewise absent and produces the same result under those controls. A KV dtype mismatch collides but never completes its load rather than silently returning foreign state (\S\ref{sec:accuracy}); vLLM binds KV dtype, so this dimension has no counterpart there.

\subsection{Correctness Impact (How Bad Is It?)}
\label{sec:accuracy}

The four dimensions span three severity bands: weight precision changes output text without degrading accuracy; adapter identity and weight version cause statistically significant, task- and model-dependent harm; and KV representation replaces the answer wholesale.

Precision is the surprising exception. On extraction, all sixty requests reuse foreign state and most outputs change, yet every answer matches the clean reference. On reasoning, the retriever receives the seeder's state over the cached prefix (mean 357 of about 440 tokens under vLLM), but no cell reaches significance ($p \geq 0.125$). Sentiment agrees: 0.935 vs.\ 0.940 at 0.5\,B and 0.970 vs.\ 0.970 at 7\,B over 200 reviews, with comparable results at 14\,B and 32\,B under vLLM.

Adapter identity, by contrast, corrupts the answer. Over 100 short reviews, an inverted adapter falls from 94\% to 64\% when served the standard adapter's KV state (32 harmful flips vs.\ 2 reversed, $p<0.001$), and reuse can return the other tenant's answer outright. On 200 IMDB reviews via vLLM runtime LoRA, accuracy falls from 0.945 to 0.715 with 28.5\% relabelled to the seeder's policy ($p<10^{-9}$). Harm depends on whether the cached prefix contains the discriminative content rather than document length. Merged-weight models show a comparable effect (0.94 to 0.60).

\begin{table}[t]
\captionsetup{skip=3pt}
\caption{Weight-version reuse across model sizes and stacks, sentiment task. Source label counts the outputs that carry the source version's answer.}
\label{tab:versionscale}
\centering
\small
\setlength{\tabcolsep}{1.5pt}
\renewcommand{\arraystretch}{0.95}
\begin{tabular*}{\columnwidth}{@{\extracolsep{\fill}}llrrrrr@{}}
\toprule
& & & \multicolumn{2}{c}{Accuracy} & Source & \\
\cmidrule(lr){4-5}
Stack & Model & $N$ & Isolated & Cross-prov. & label & $p$ \\
\midrule
vLLM & Qwen2.5-0.5B & 300 & 0.893 & 0.750 & 74/300 & $<0.001$ \\
vLLM & Qwen2.5-7B & 300 & 0.947 & 0.900 & 30/300 & $0.004$ \\
vLLM & Qwen2.5-14B & 300 & 0.957 & 0.910 & 27/300 & $0.003$ \\
vLLM & Qwen2.5-32B & 300 & 0.957 & 0.943 & 17/300 & $0.39$ \\
SGLang & Qwen2.5-0.5B & 400 & 0.902 & 0.865 & --- & $0.014$ \\
SGLang & Qwen2.5-7B & 400 & 0.943 & 0.938 & --- & $0.63$ \\
\bottomrule
\end{tabular*}
\end{table}

Weight version behaves like an adapter, with task-dependent magnitude. Under vLLM, two opposite-policy fine-tunes serve as versions merged into full checkpoints. Across 300 reviews per model from 0.5\,B to 32\,B, reuse stays near 92\% of prompts while accuracy loss falls from 14.3 points at 0.5\,B to 1.4 at 32\,B, where the difference is unresolved ($p=0.39$, Table~\ref{tab:versionscale}). On SGLang, sentiment attenuates similarly (3.7 points at 0.5\,B to an unresolved 0.5 at 7\,B), whereas reasoning rises from an unresolved 3.0 to a significant 8.0 at 7\,B. The earlier release reproduces this split at 7\,B on the other platform. At 14\,B both deltas are small and unresolved even at four times the sample size, so attenuation is not a power artifact. Harm does not grow with model size on either graded task.

Harm is robust to training randomness: three independently trained 7\,B pairs all decline significantly, with net harmful flips in every case. The effect is symmetric (standard retrieving inverted: 0.97 to 0.90, $p<0.02$) and reproduces across families (Llama-3.1-8B, Mistral-7B). Released checkpoints also collide: a base model seeds and its instruction-tuned sibling retrieves foreign state on 18 of 20 prompts, yet scores 0.94 both ways. On GSM8K, reuse fires on 96\% of problems but accuracy moves nonsignificantly (0.76 to 0.71), because the short question caches about seventy tokens and the answer-determining reasoning is recomputed. Harm appears only where versions disagree on the answer the cached prefix determines.

At the other extreme, changing the KV pool element type replaces the answer outright. Sentiment falls to about 0.47 across 0.5\,B--14\,B (from 0.945--0.975), with 298 of 299 paired flips harmful. On reasoning, collided accuracy is zero at every size and family tested; the widening gap reflects a rising isolated reference rather than a worsening collision. A KV dtype mismatch depends on storage policy: under the waiting default, the foreign-page load never completes, a liveness failure; under the permissive policy, the engine declines the entry and matches the isolated reference. An own-store control loads 19,200 tokens and remains healthy, which confirms that the permissive policy stays safe only because it declines to load the foreign entry at all.

\subsection{Causal Analysis (Why Does It Happen?)}
\label{sec:causal}

Comparing tasks and dimensions confounds the task, the data and the adapters, so two controlled experiments isolate the cached state itself as the direct cause of the loss.

Holding reviews, model and adapters fixed, a dose-response experiment varies the foreign fraction of the cached prefix while the retriever recomputes the remainder. Accuracy falls from 0.92 with no foreign state to 0.76 at one quarter (33 tokens), then remains steady through one half (95), three quarters (158) and the whole prompt (220), while emission of the seeder's label rises from 8\% to 25\%. Because only the foreign fraction changes, cached state is causal, and the early saturation shows that a partial foreign prefix carries the label-determining framing. The zero-to-quarter step is significant (17 flips vs.\ 1, $p<10^{-3}$), whereas a position control that fixes the foreign token count and moves sentiment-bearing content between cached and recomputed portions is only marginal ($p{=}0.06$), confirming that foreign state rather than its position causes the harm.

A complementary threshold experiment makes the firing condition mechanical: the connector stores fixed-size chunks, so prefixes below one chunk never collide, whereas those at or above it collide whenever shared. Sweeping a system prompt across two adapters, a sub-chunk prefix serves no cross-adapter state, while prefixes of roughly 100--1,000 tokens serve tens to several hundred foreign tokens, one chunk at a time. SGLang exhibits the same threshold at one storage page. Production shared prefixes span hundreds of tokens, so collisions fire deterministically whenever such a prefix crosses a provenance boundary.

\subsection{Deployment Incidence (How Likely in Realistic Deployments?)}
\label{sec:fleet}

The cross-provenance rate in production lies between two tight bounds and rises toward the ceiling precisely when the cache tier is doing the job it was deployed for.

We bound the rate using eight production-style system prompts shared by \texttt{bfloat16} and \texttt{fp8} tenants, each paired with a real user query. A two-phase seed-then-replay construction serves cross-provenance state to 167 of 200 requests (83.5\%), an upper bound because every replayed prefix is warm under the other class; an interleaved stream with one always-warm worker per class falls to 3 of 200 (1.5\%), as workers reach the store only on a cold first touch. The realistic rate depends on the fleet's reuse distribution, which is skewed and capacity-sensitive~\cite{wang2025kvcachewild}; because existing traces lack provenance, these bounds are constructed rather than derived from a production trace.

Fleet diversity and routing shape the realized rate. On a live four-worker fleet (two per precision) sharing one store over 600 requests and 32 system prompts, load-balanced routing produces 10\% cross-provenance reuse while sticky routing almost never does. As prefix diversity rises, the load-balanced rate grows from 3\% at 8 prefixes to 8\% at 32 and 16\% at 128 because more prefixes overflow local caches and send cold touches to the store; sticky routing remains at zero until diversity exceeds local capacity, reaching 4\% at 128. A joining or restarting worker reads the store on every request from a cold local cache and receives foreign state on all of them when another provenance has populated the store, a routine condition during autoscaling or rolling updates (\S\ref{sec:recovery}).

\subsection{Confidentiality Impact (Can It Leak Information?)}
\label{sec:confidentiality}

The same omission, seen through a confidentiality lens, turns the shared tier into a cross-instance side channel that defeats the isolation control the engine already deploys.

With the tier disabled in one instance, a request with a different salt or none misses a salted victim's blocks, confirming documented engine-local isolation. Enabling the tier serves that state across instances to an unsalted or differently salted request, whose \needload{} confirms a store load on both connectors tested end to end. The bypass reproduces across seven families from 0.5\,B to 14\,B (Qwen 0.5\,B--14\,B, Llama-3.1-8B, Mistral-7B, Phi-3.5-mini, Yi-1.5-6B, Falcon3-7B and SmolLM2-1.7B), always serving the salted victim's state to the unsalted or differently salted attacker while preserving same-salt reuse in every family.

Composed end to end, the bypass enables a concrete attack: an attacker timing
$N$ candidate prompts against a salted victim identifies the victim's prompt
perfectly up to ten candidates and at 93\% for $N{=}100$, the exact scenario
the salt was introduced to prevent. Saturating load on the store closes the
channel, while a faster transport widens the gap.

One dimension admits a release-level fix. SGLang~0.5.20 isolates salted requests in both tiers: a second sharing domain presenting identical tokens consumes nothing from either and receives the clean-store output. On 0.5.8.post1 the same test reuses the first domain's 1\,024 cached tokens. Salt behaviour is therefore a release-level property. Meanwhile, none of the audited vLLM connectors binds the salt, which is opt-in and off by default, and one commercial serverless offering documents a fleet-wide shared cache that offers no isolation between separate customer accounts.

Across both stacks these are one gap and not five: every failure above follows
from a shared key that carries token content and coarse model identity while
the engine, and only the engine, holds the provenance that key would need.
Binding one dimension, as a release may, leaves the others open. What the
interface lacks is a contract for that provenance.

\section{KV Provenance Contract and Design}
\label{sec:defense}

Every failure in \S\ref{sec:collisions} stems from a shared key that encodes tokens but not their production context. This section formalizes the missing guarantee as a parametric contract over a declared dimension registry, gives a canonical realization extensible without connector-specific key changes, binds request- and worker-level provenance, maps the coverage boundary, and defines a differential omission test for any dimension that is left unbound.

\subsection{The Provenance Contract}
\label{sec:contract}

\parhead{Provenance dimensions.}
Let $\mathcal{D} = \{d_1, \ldots, d_n\}$ be a finite \emph{dimension registry},
where each dimension~$d_i$ has a value domain $\mathrm{Dom}(d_i)$ and an
\emph{identity function} $\mathrm{id}_i : \mathrm{Dom}(d_i) \to \mathcal{I}$ that maps
concrete runtime values to identities in a common identity space~$\mathcal{I}$.
A dimension is \emph{well-formed} when its identity function is injective and
cross-worker stable: for all workers $w, w'$ and values $v \in \mathrm{Dom}(d_i)$,
$\mathrm{id}_i^{\,w}(v) = \mathrm{id}_i^{\,w'}(v)$.
A \emph{provenance configuration} over~$\mathcal{D}$ is a valuation
$c \in \prod_{d \in \mathcal{D}} \mathrm{Dom}(d)$. Write $KV(x,c)$ for the
state a worker produces for tokens~$x$ under~$c$, and $\mathrm{Key}(x,c)$ for the shared
key it derives.

\begin{definition}[KV Provenance Contract $\mathcal{C}(\mathcal{D})$]
\label{def:contract}
Given a well-formed dimension registry~$\mathcal{D}$, a key function $\mathrm{Key}$
satisfies the \emph{KV provenance contract} $\mathcal{C}(\mathcal{D})$ iff for all token
sequences~$x$ and configurations $c_1, c_2$ over~$\mathcal{D}$:

\smallskip\noindent
\emph{(i) Computational invariant.}
Configurations producing different state derive different keys:
\begin{equation}
KV(x,c_1) \neq KV(x,c_2)
  \;\Longrightarrow\;
  \mathrm{Key}(x,c_1) \neq \mathrm{Key}(x,c_2).
\label{eq:computational}
\end{equation}

\noindent
\emph{(ii) Permission invariant.}
Configurations not permitted to share derive different keys:
\begin{equation}
\neg\,\mathrm{ShareAllowed}(c_1,c_2)
  \;\Longrightarrow\;
  \mathrm{Key}(x,c_1) \neq \mathrm{Key}(x,c_2).
\label{eq:permission}
\end{equation}
\end{definition}

The contract is parametric in~$\mathcal{D}$: it specifies what bound dimensions must
satisfy, not which dimensions exist. It is stated over \emph{declared}
dimensions, because binding stable identities is what an implementation can do, while
deciding equivalence of arbitrary KV-generating functions is not.

\subsection{Canonical Realization and Extensibility}
\label{sec:realization}

Our implementation realizes $\mathcal{C}(\mathcal{D})$ through a \emph{canonical
encoding}: an injective, length-prefixed encoder
$\mathrm{enc} : \prod_{d \in \mathcal{D}} \mathcal{I} \to \{0,1\}^*$ paired with a
collision-resistant hash~$H$, and the two together yield the shared key
\begin{equation}
\mathrm{Key}(x,c)
  \;=\;
  f\!\bigl(x,\;
  H(\mathrm{enc}(
    \mathrm{id}_1(c(d_1)), \ldots, \mathrm{id}_n(c(d_n))
  ))\bigr),
\label{eq:realization}
\end{equation}
where $f$ is the connector's native key function. The realization assumes that,
for fixed~$x$, $f(x,\cdot)$ preserves distinct provenance digests; if it hashes
them again, its collisions are an additional failure event. Hash-based
identities and keys approximate the ideal contract only up to the probability that the hash itself collides.

\begin{proposition}[Dimension Extension]
\label{prop:extension}
Let\/ $\mathrm{Key}$ realize $\mathcal{C}(\mathcal{D})$ via a canonical encoding, and assume $f(x,\cdot)$ preserves distinct provenance digests. For any well-formed new dimension $d_{n+1}$, extending\/ $\mathrm{enc}$ by appending $\mathrm{id}_{n+1}$ satisfies the extended contract $\mathcal{C}(\mathcal{D} \cup \{d_{n+1}\})$, except with negligible hash-collision probability under the assumed hash security.
\end{proposition}
\begin{proof}[Proof sketch]
Injectivity of length-prefixed encoding is preserved under field extension: if two
configurations differ on any dimension, their encodings differ. Conditioned on
no hash collision, their digests differ, and the assumption on~$f$ preserves
that distinction in the shared key. The only probabilistic step is excluding
hash collisions; collision resistance does not make~$H$ injective.
Cross-worker stability follows from the well-formedness of the extended registry.
\end{proof}

Extension requires only one registry entry and populated field through the existing serialization; connector-specific key logic remains unchanged because connectors consume the digest rather than individual fields. The current instantiation binds adapter, salt, weight precision and weight version, and additionally binds the resolved KV representation when the native shared key does not already encode it.

The contract also assigns responsibility: the engine alone holds request context and must supply its known provenance, while the connector alone builds the key and must bind that provenance at every access site without inferring safe reuse from equal tokens. We carry provenance as one descriptor with request- and worker-level halves, both serialized by the same length-prefixed encoder over one registry and reduced by the same digest, so neither is a special case.

\subsection{Binding Mechanism}
\label{sec:bindingmech}

The canonical realization reduces the contract to computing a digest over dimension
identities, but in a running system those identities must reach every site that derives a key.

\parhead{Worker-level digest.}
Weight precision, a weight-version identifier and the key-value representation are fixed
for a worker's lifetime, so we digest them once at startup and fold the result into the
model identity the connector already carries. The values digested are
the ones the worker resolved, not the ones it was launched with. An engine that accepts
\texttt{auto} for a dtype decides the element type at pool allocation, so identically
launched workers can resolve differently. Binding the launch string would therefore merge configurations
that compute different state and split ones that agree. The
SGLang integration reads the resolved value and keeps the launch string as metadata that never
reaches a key.

\parhead{Request-level tags.}
The adapter identity and salt vary per request and travel with it.
For the salt we bind the value the request carries. For the adapter we bind an identity of
the object rather than a process-local handle: on vLLM, the identity is the adapter name,
stable across workers; on SGLang, the integration fingerprints the
adapter's contents (every tensor's name, dtype, shape and bytes, together with the
semantically meaningful configuration fields), while the engine's per-adapter
handle indexes a registry of those fingerprints. Content
identity ensures that the same bytes under a second path share state, while
replaced bytes at the same path do not. The binding is
inert for a request carrying neither adapter nor salt, so such a request keys exactly as
unpatched and the worker-level digest is undisturbed. The two halves are independent of one another, and they compose without interfering with each other.

\parhead{Cross-site consistency.} Both halves compose into a single descriptor, and that descriptor has to be available at every site that touches the store, or the computational invariant~\eqref{eq:computational} holds at one site and fails at the other. The shared store is accessed at two sites: the \emph{scheduler}, which looks up state before prefill, and the \emph{worker}, which stores it after. The scheduler holds the request record with the adapter name and the salt, while the worker holds the weights and the computed state but its copy of that record carries neither reliably. Injecting provenance at one site only breaks same-provenance reuse silently, because the store site then derives a key the lookup site cannot reproduce. In the implementation the worker's request record carries the salt and the adapter identity, so the lookup site and the store site both derive their keys from the same descriptor value.

\parhead{Connector integration.}
Given a consistent descriptor, the remaining step is embedding it into each connector's
native key format. What reaches the connector is the descriptor's canonical digest rather
than the field values, so connectors share the provenance encoding, not the key itself.
LMCache carries the worker digest in the model identity and the request digest in one
namespaced tag, while vLLM's filesystem and Mooncake stores fold the worker material into
their \texttt{md5} and \texttt{blake2b} inputs. SGLang HiCache concentrates the choice in
one place: every storage chain it builds, for lookup, write, promotion and restart alike,
descends from a single namespace seed, so the integration replaces that seed and reaches
all four without a per-site change. The seed that digest replaces on SGLang is the one the
engine already uses to separate sharing domains, so the native salt keeps working and any
non-adapter component of the key's extra field travels through the descriptor verbatim.
Within each integration, lookup and store bind the same descriptor, so site-specific encodings cannot cause silent disagreement.

One connector's wire format constrains how the digest is encoded rather than how much
provenance is bound. LMCache packs its key into a fixed 150-byte field, which a
128-bit hex digest appended to an unbounded model path would overrun.
Because the digest width is a security parameter, the encoding keeps 128 bits but shortens
the representation: twenty-two base64 characters for the digest and a fixed-length hash over
the model identity and worker provenance. A representative key is 89 characters regardless of path length, so the budget holds by construction.

\subsection{Differential Omission Testing}
\label{sec:omissions}

A parametric contract is testable per dimension. For each $d_i \in \mathcal{D}$, the
\emph{differential omission test} holds token sequence~$x$ and all dimensions
except~$d_i$ fixed, varies~$d_i$ to a distinct value~$v'$, and reports an omission iff
\begin{equation}
\small
KV(x, c) \neq KV(x, c[d_i \mapsto v'])
\;\;\wedge\;\;
\mathrm{Key}(x, c) = \mathrm{Key}(x, c[d_i \mapsto v']).
\label{eq:omission}
\end{equation}
An omission at~$d_i$ directly falsifies the computational
invariant~\eqref{eq:computational}. The permission
invariant~\eqref{eq:permission} is tested analogously, replacing KV inequality with
$\neg\,\mathrm{ShareAllowed}$.

The offline form compares real attention state on one machine: a model runs on CPU in
float32 and its prefill keys and values are compared layer by layer, so two runs of one
configuration agree bitwise. On Qwen2.5-0.5B it reports all five supplied dimensions as
omissions under the unpatched content-only key and one, the deliberately held-out rotary
base, under the bound key, the largest absolute difference between corresponding cached
tensors being 1.03 for the adapter, 6.54 for precision and 2.14 for the weight version,
while the sharing domain changes no state and is decided by the permission oracle.
Disabling one binding reopens exactly that dimension and no other. Registering the
held-out dimension routes it through the existing encoder and key derivation, one registry
entry and no new key-derivation code, after which the checker reports no omission over
this space: it evaluates the declared candidates it is supplied and is not a discovery
oracle for an unenumerated configuration space. \S\ref{sec:omissionlive} runs the same test
inside a serving engine, where the state is what two workers wrote into the shared store
and the keys are the names the store gave those entries.

\section{Evaluation of Provenance Binding}
\label{sec:eval}

This section evaluates whether provenance binding blocks unsafe reuse while preserving legitimate sharing (\S\ref{sec:blocking}), whether its dimensions act independently (\S\ref{sec:ablation}, \S\ref{sec:omissionlive}), and how it performs under rolling updates, cold starts and practical resource costs (\S\ref{sec:recovery}, \S\ref{sec:costs}).

We evaluate two production stacks: vLLM~0.9.2 with LMCache and its filesystem store, and SGLang~0.5.20 with HiCache, using the models, datasets and seeder/retriever protocol of \S\ref{sec:methodology}. Bound and unbound conditions otherwise match; workers restart with cold local caches, and every timed retrieval is a confirmed store load. Across all experiments, the binding eliminates every identified unsafe reuse case, preserves legitimate sharing, adds no resolvable hit-path latency, and incurs storage in proportion to the provenance diversity that a deployment actually presents.

\begin{figure*}[t]
  \centering
  \includegraphics[width=\textwidth]{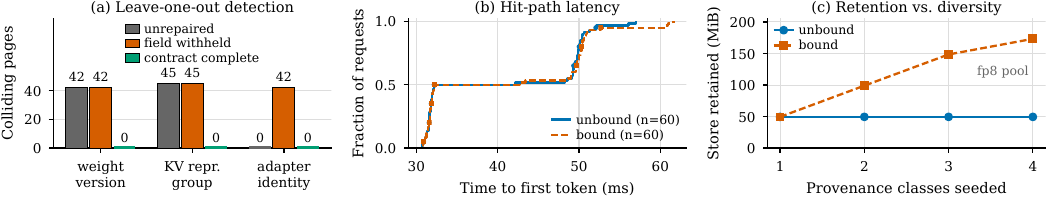}
  \caption{Provenance binding on SGLang~0.5.20 with HiCache, 0.5\,B. In (a)~leave-one-dimension-out, bars count store pages with differing bytes under an identical key. (b)~Time to first token on confirmed store hits, two runs of 30 requests per arm ($n=60$ per arm). (c)~Store retention against the number of provenance classes seeded.}
  \label{fig:repair}
\end{figure*}

\subsection{Correctness of Provenance Binding}
\label{sec:blocking}

The binding covers five dimensions: four \emph{computational} dimensions that change cached state (adapter identity, weight precision, weight version and key-value representation), and one \emph{permission} dimension, the tenant salt, which can prohibit sharing even when state is identical.

For each dimension, a 0.5\,B worker pair seeds the same tokens under one configuration and retrieves them under another through LMCache and vLLM's filesystem store. Binding reduces all five cross-configuration cases to zero foreign tokens while preserving full-prefix same-provenance hits (624 tokens on the filesystem store and 768 under LMCache); disabling it restores every cross case to a hit, confirming that the binding causes the misses.

The dimensions also compose: a seeder differing in precision and weight version reuses under the original key but misses under the bound key, reducing reuse from 91\% to zero. When all five dimensions differ, all twelve 360-token prompts receive foreign state under the original key and none under the bound key.

\parhead{Paired matrix on SGLang.} 
Table~\ref{tab:sglmatrix} evaluates thirteen SGLang provenance cases in bound and unbound arms at 0.5\,B, using forty documents and 5,376 prefix tokens per read. Each case runs five separate-process phases: a no-store reference, source seed, cross-configuration read, reader seed and warm read verifying the store. A bound cross-read is safe only with zero foreign tokens and a completed response matching the reference; empty or failed responses do not qualify. Twenty-five of twenty-six arms pass, except the unbound run-time adapter case, where identical adapter bytes receive different engine handles across processes.

\begin{table}[t]
\captionsetup{skip=3pt}
\caption{Paired SGLang~0.5.20 HiCache matrix at 0.5\,B. \emph{Foreign}: reads the other configuration's state; \emph{safe miss}: zero foreign tokens and a clean-reference response; \emph{hit}: same-provenance reuse; \emph{isolated}: the native key separates configurations. $\dagger$~incomplete response; $\ddagger$~no self-reuse.}
\label{tab:sglmatrix}
\centering
\small
\setlength{\tabcolsep}{2pt}
\renewcommand{\arraystretch}{0.95}
\begin{tabular*}{\columnwidth}{@{\extracolsep{\fill}}lll@{}}
\toprule
Case & Unbound & Bound \\
\midrule
Weight version & foreign & \textbf{safe miss} \\
Weight precision & foreign & \textbf{safe miss} \\
Key-value element type & foreign$^\dagger$ & \textbf{safe miss} \\
Key-value dtype & load failure & \textbf{safe miss} \\
Version and element type & foreign$^\dagger$ & \textbf{safe miss} \\
Adapter replaced in place & foreign & \textbf{safe miss} \\
\addlinespace[2pt]
Same provenance & hit & hit \\
Adapter across restart & hit & hit \\
Adapter at another path & no reuse & \textbf{hit} \\
Adapter, run-time handles & no reuse$^\ddagger$ & \textbf{hit} \\
\addlinespace[2pt]
Tenant salt & isolated & isolated \\
Different adapter content & isolated & isolated \\
Adapter, different salt & isolated & isolated \\
\bottomrule
\end{tabular*}
\end{table}

Five unbound cases load the full 5,376-token foreign prefix, with incomplete responses in the two element-type cases marked in Table~\ref{tab:sglmatrix}. The sixth case, key-value dtype, fails to complete the mismatched load; it is a load failure, not a completed foreign-prefix read. Binding turns all six into safe misses with zero foreign tokens and responses matching the clean reference, while their own warm reads still load all 5,376 tokens. Binding also recovers legitimate sharing lost by the engine: identical adapter content at another path or under different run-time handles reuses the full prefix only when bound, because adapter identity comes from content rather than engine-internal handles.

\subsection{Ablation Study}
\label{sec:ablation}

To distinguish independent contributions from an all-or-nothing effect, we combine the two request-level dimensions, adapter and salt, in one workload and enable their bindings separately. A seeder writes each prefix under one adapter--salt pair, and a retriever requests the same tokens under four provenance combinations (Table~\ref{tab:leaveoneout}).

\begin{table}[t]
\captionsetup{skip=3pt}
\caption{Leave-one-binding-out on a combined adapter--salt workload. Cells report prefix tokens loaded from the shared store (\needload{}), including legitimate same-provenance reuse. Each binding zeros its single-dimension cross case; both zero every cross case. The 16-token residuals are discussed in the paragraphs that follow.}
\label{tab:leaveoneout}
\centering
\small
\setlength{\tabcolsep}{2pt}
\renewcommand{\arraystretch}{0.95}
\begin{tabular*}{\columnwidth}{@{\extracolsep{\fill}}lrrrr@{}}
\toprule
& Same & \multicolumn{3}{c}{Cross-provenance} \\
\cmidrule(l){3-5}
Binding & provenance & Adapter & Salt & Both \\
\midrule
none          & 656 & 640          & 672         & 672 \\
adapter only  & 656 & \textbf{0}   & 672         & 16 \\
salt only     & 656 & 640          & \textbf{0}  & 16 \\
both          & 656 & \textbf{0}   & \textbf{0}  & \textbf{0} \\
\bottomrule
\end{tabular*}
\end{table}

With neither binding active, the original key leaks 640 cross-adapter and 672 cross-salt tokens. Adapter-only binding zeros the former but leaves 672 cross-salt tokens; salt-only binding does the converse, leaving 640; and both zero every cross case while preserving same-provenance reuse at 656. These single-dimension cases demonstrate that the two bindings contribute independently of each other.

When both dimensions differ, either single binding leaves one 16-token chunk. The four conditions share a store, so these counts can include template state written by earlier retriever requests. In an adapter-only fresh-store control with the seeder as the only prior writer, the cross-provenance read loads zero tokens; subsequent requests under the retriever's own adapter show 16-token template hits with zero store loads. This supports template self-reuse as the explanation for the residual, but does not independently establish the salt-only residual's provenance. We therefore report the observed store-load counts, and we do not classify either of the two residuals as state carried over from a foreign provenance.

\subsection{Omission Detection}
\label{sec:omissionlive}

The ablation checks what a retriever receives; a complementary store-side test asks whether configurations producing different cached state share a key, directly testing whether the binding omits any required dimension.

Two workers differing in one dimension write the same documents into one store, and a checker flags matching keys whose page bytes differ, indicating an omitted binding. Figure~\ref{fig:repair}(a) reports these collisions for the unbound, leave-one-dimension-out and complete bindings. Withholding weight-version identity collides on 42 of 45 pages, key-value representation (element type, dtype and layout, which change together with launch dtype) on all 45, and adapter identity on 42 of 43. Restoring each withheld dimension reduces its observed collision count to zero. This validates the tested bindings for the supplied configuration changes and pages; it does not establish completeness over untested values or an unenumerated configuration space.

\subsection{Deployment Validation}
\label{sec:recovery}

Rolling updates and cold starts make unsafe reuse most likely: old-version workers warm the store before new ones read it, while newly launched workers consult the shared store on every request until their local caches warm.

\parhead{Paired rollout.} An old-version worker warms the store before old- and new-version traffic is interleaved at new-version fractions from one quarter to all, with a fixed seed giving the bound and unbound arms the same request stream. Binding eliminates stale serving at every ratio. At one-quarter new traffic, the unbound stale-fetch rate is 0.24 and accuracy is 0.87, versus zero and 0.92 when bound; at three quarters, accuracy is 0.89 versus 0.94. The effect tracks the stale rate and is largest early, while old-version state dominates the store.

\parhead{Cold-start recovery.} A worker joining or restarting mid-rollout has a cold local cache and reads the shared store on every request. If the old version warmed that store, all 100 requests receive stale state and accuracy is 0.73; binding restores it to 0.92 by preventing the cold worker from reading old-version entries. Cold starts during rollouts therefore yield the binding's largest correctness benefit.

\subsection{Overhead Analysis}
\label{sec:costs}

A binding must also be affordable at serving scale. It adds a one-time startup digest for worker-level dimensions (weight precision, weight version and key-value representation) and a per-request digest for request-level dimensions (adapter identity and tenant salt). We measure hit-path latency, sustained throughput, startup time and storage.

\parhead{Hit-path latency.} On vLLM, across three counterbalanced repeats of 50 requests per condition, median time to first token changes by $-0.30$~ms (36.7 vs.\ 36.4~ms) and p90 by $-0.5$~ms, within the 3.0~ms between-repeat spread. On SGLang, four runs in bound--unbound--unbound--bound order give two runs per arm, each with 30 confirmed 4,352-token store loads ($n=60$ per arm), and produce modes near 31.6 and 50~ms (Figure~\ref{fig:repair}b). The pooled mean shifts by $+0.34$~ms, and the lower and upper modes by $+0.03$ and $+0.64$~ms, against 3.0~ms of between-repeat variation. Neither engine shows a resolvable penalty at our measurement scale: the request-level digest takes 12.7~$\mu$s, versus tens of milliseconds for a store round trip.

\parhead{Sustained throughput.} With two same-provenance 7B vLLM workers cross-seeding and retrieving at concurrency 25, all 900 requests per condition are confirmed 384-token store fetches on both workers. Across three counterbalanced repeats, the bound and unbound distributions overlap within between-repeat variation, so no penalty in sustained throughput is resolvable at this concurrency.

\parhead{Startup costs.} Fingerprinting a 40~MB adapter takes 1.06~s cold and 0.060~s warm per load; weight-version metadata fingerprinting takes 0.1~ms. Content-hashing a 0.5B checkpoint takes 1.35~s and scales with checkpoint size. These one-time costs amortize over the worker's lifetime.

\parhead{Storage growth.} Binding separates entries that the original key unsafely deduplicates, so storage scales with provenance diversity. For twelve prefixes on vLLM's filesystem store, unbound size remains 214~MB, while bound size grows across four classes from 214 to 428, 642 and 856~MB. For twenty documents on SGLang's file backend, the corresponding counts are a flat 66 pages unbound and 66, 132, 198 and 264 pages bound (Figure~\ref{fig:repair}c). The fourth class adds 24.75 rather than 49.5~MiB because it uses an \texttt{fp8} pool, so footprint reflects element width. Growth is linear and is the minimum needed to keep provenance classes separate.

\parhead{Eviction under fixed budget.} Under a fixed budget, provenance separation increases eviction pressure. On vLLM, unbound deduplication retains one fitting copy, so latency and throughput remain flat. With binding, exceeding capacity causes thrashing: time to first token rises from 0.024 to 0.065~s and throughput falls from 42 to 15~req/s before plateauing. This is the cost of eliminating unsafe deduplication; provisioning for expected provenance diversity avoids it.

\section{Discussion}
\label{sec:discussion}

\parhead{Interface implications.}
Each layer is locally correct: the engine binds provenance into its key, while the shared tier uses content addressing to deduplicate. The failure lies at the seam, where the tier re-derives a key from tokens without the engine's context. Any request policy, such as a rate class, region or data-residency tag, is likewise lost unless carried through. We therefore bind provenance into the key rather than detect collisions after they produce wrong output.

\parhead{Scope and Generality.}
The failure belongs to the interface, not any single connector or stack: all audited connectors key alike and both characterized stacks reuse across provenance boundaries (\S\ref{sec:reuse}). Omitted dimensions vary: SGLang~0.5.20 binds the tenant salt whereas 0.5.8.post1 does not, and vLLM binds the key-value dtype. Adapter identity offers the sharpest contrast: every audited vLLM connector omits it, while SGLang~0.5.20 uses a name-and-path handle that separates named adapters but not versions under one name. Because harm varies, we grade correctness only where ground truth exists (\S\ref{sec:accuracy}). Incidence uses a controlled two-phase workload because production traces lack provenance labels. Throughput and eviction are measured on vLLM~0.9.2.

Our measurements use single-host workers accessing a standalone cache server, not a multi-host fleet. Generalization depends on architectural compatibility (cross-version reuse requires cached tensors to fit the retriever and versions to disagree on the cached prefix's answer, so opposite-policy checkpoints show substantial harm while agreeing ones do not) and the identity chosen per dimension, which the differential checker validates against exact KV state.

\section{Related Work}
\label{sec:related}

\parhead{LLM serving and shared KV caching.}
vLLM~\cite{kwon2023efficient} introduced PagedAttention and prefix caching; SGLang~\cite{zheng2024sglang} uses a radix tree. LMCache~\cite{cheng2025lmcache} and Mooncake~\cite{qin2025mooncake} share KV across instances through external stores, while CachedAttention~\cite{gao2024cachedattention} retains inactive sessions hierarchically. Disaggregated serving~\cite{zhong2024distserve,hu2024tetriinfer,patel2024splitwise} makes cross-instance KV transfer first-class; CacheGen~\cite{liu2024cachegen} compresses and streams it, MemServe~\cite{hu2024memserve} and D\'{e}j\`{a}Vu~\cite{strati2024dejavu} manage distributed pools, and DroidSpeak~\cite{liu2024droidspeak} reuses KV across models sharing an architecture. They assume matching tokens imply matching KV state, which fails across workers that differ in adapter, in precision or in weight version.

\parhead{When is cached state reusable.}
Prompt Cache~\cite{gim2024promptcache} reuses state for modular non-prefix segments under an author-declared schema; CacheBlend~\cite{yao2025cacheblend} reuses non-prefix chunks and recomputes a few tokens to repair cross-attention omitted by reuse. Both address position within one engine and model configuration, assuming producer and consumer configurations match. Our question is orthogonal: state can be positionally valid yet produced by another configuration.

\parhead{Cache timing side channels and isolation.}
InputSnatch~\cite{inputsnatch}, EarlyBird~\cite{earlybird} and PromptPeek~\cite{promptpeek} recover prompts by timing prefix-cache hits within the victim's instance, which the salt is meant to prevent. We show that a shared tier discarding the salt extends this leak to a remote attacker on another instance. SpliceLeak~\cite{spliceleak} studies a non-prefix-fusion side channel in the same tier but requires co-location and does not concern the salt. KVGov~\cite{kvgov} binds a per-principal keyed hash into the block hash; PrefixWall~\cite{prefixwall}, SafeKV~\cite{safekv} and CachePrune~\cite{cacheprune} isolate or selectively share prefix state within one engine. None considers a cross-instance tier that re-derives its key without the context these defenses assume.

\parhead{KV-state integrity, adapters and models.}
Cache manipulation steers outputs through writes~\cite{whosenarrative} or faults in shared blocks~\cite{bitflip}, requiring adversarial write or fault capability; our collisions instead serve correctly computed state under the wrong provenance. Multi-adapter systems~\cite{sheng2024slora,chen2024punica,wu2024dlora} batch LoRA requests and migrate adapters between replicas, requiring stable cross-worker identities while relying on per-instance isolation that shared tiers do not replicate. The issue is recognized in practice: an adapter-aware LMCache key was closed as not planned~\cite{lmcacheissue}, a vLLM issue records the conflict~\cite{vllmissue}, and recent work exploits cross-adapter reuse while acknowledging naive sharing is lossy~\cite{alora}. We unify unintentional adapter, precision, version and KV-representation reuse as provenance-blind reuse, demonstrate a cross-tenant integrity failure, and close all four with one binding.

\section{Conclusion}
\label{sec:conclusion}

This paper identifies provenance-blind reuse as a structural composition gap at the engine--connector boundary in shared KV caching: every audited connector omits provenance that the engine tracks locally, causing cross-configuration collisions that reduce task accuracy by up to 30 points or to zero. We formalize the missing guarantee as the KV provenance contract, parametric over a declared dimension registry so that coverage extends without connector-specific key changes, and pair it with a differential checker that makes omitted bindings directly testable. Provenance binding, implemented in vLLM and SGLang across three cache paths, eliminates unsafe reuse on all five bound dimensions while preserving same-provenance sharing at no resolvable hit-path latency penalty.

\bibliographystyle{ACM-Reference-Format}
\bibliography{references}

\end{document}